\documentclass[a4paper,11pt]{article}
\usepackage[latin1]{inputenc}
\usepackage[english]{babel}

\usepackage{latexsym}
\usepackage{amsfonts}
\usepackage[latin1]{inputenc}
\usepackage{graphicx}
\usepackage{amsmath}
\usepackage{amsthm}

\usepackage{amssymb}
\usepackage[all]{xy}
\usepackage[normalem]{ulem}
\usepackage{lscape}
\usepackage{color}

\usepackage{enumitem}
\usepackage{hyperref}

\usepackage{times} 
\usepackage{cite}

\usepackage{mathabx}

\usepackage{adjustbox}

\usepackage{MnSymbol}

\newcounter{desccount}
\newcommand{\descitem}[1]{%
  \item[#1] \refstepcounter{desccount}\label{#1}
}
\newcommand{\descref}[1]{\hyperref[#1]{#1}}

\newtheorem{defi}{Definition}[section]
\newtheorem{teo}{Theorem}[section]
\newtheorem{exem}{Example}[section]
\newtheorem{obs}{Remark}[section]
\newtheorem{prop}{Proposition}[section]
\newtheorem{lem}{Lemma}[section]
\newtheorem{coro}{Corollary}[section]

\begin{document}


\title{($S$,$N$,$T$)-Implications}

\author{Fernando Neres\\ fernandoneres@ufersa.edu.br \\ Departamento de Ci\^encia e Tecnologia - DCT \\
Universidade Federal Rural do Semi-\'Arido -- UFERSA\\
Cara\'ubas, Rio Grande do Norte, Brazil\\[0.4cm] Benjam\'in Bedregal\\bedregal@dimap.ufrn.br\\Departamento de  Inform\'atica e Matem\'atica Aplicada - DIMAp \\ Universidade Federal do Rio Grande do Norte - UFRN\\
Natal, Rio Grande do Norte, Brazil\\[0.4cm] Regivan H. N. Santiago\\regivan@dimap.ufrn.br\\Departamento de  Inform\'atica e Matem\'atica Aplicada - DIMAp \\ Universidade Federal do Rio Grande do Norte - UFRN\\
Natal, Rio Grande do Norte, Brazil}

%
%
%
%
%
%

\maketitle

\begin{abstract}
In this paper we introduce a new class of fuzzy implications called ($S$,$N$,$T$)-implications inspired in the logical equivalence $p\rightarrow q \equiv \neg(p\wedge\neg q)\vee\neg p$ and present a brief study of some of the main properties that characterize this class. We  present methods of obtaining $t$-norms and $t$-conorms from an ($S$,$N$,$T$)-implication and a fuzzy negation.
\end{abstract}

%


\section{Introduction}

Just as the connective $t$-norm and $t$-conorm of fuzzy logic are generalizations of the connective conjunction ($\wedge$) and disjunction ($\vee$) of classical logic, respectively; fuzzy implication \cite{Baczynski2008} is an important connective of fuzzy logic that generalizes the connective implication of classical logic, taking truth values in the interval $[0,1]$ instead of in the set $\{0,1\} $.

Different classes of fuzzy implications have their origin from logical equivalences of the classical implication, e.g., the ($S$,$N$)-implications \cite{Baczynski2007} that generalize the material implication $p\rightarrow q\equiv\neg p\vee q$, the $QL$-implications \cite{Dimuro2017,Mas2006} that generalize the implication defined in quantum logic $p\rightarrow q\equiv\neg p\vee(p\wedge q)$ and the $D$-implications \cite{Dimuro2019a,Mas2006} that generalize the Dishkant implication $p\rightarrow q\equiv q\vee(\neg p\wedge\neg q)$ of orthomodular lattices.

In this paper, we will introduce a new class of fuzzy implications called ($S$,$N$,$T$)-implications that generalizes the logical equivalence $$p\rightarrow q \equiv \neg(p\wedge\neg q)\vee\neg p$$ where $\wedge$ is replaced by a $t$-norm, $\vee$ by a $t$-conorm and $\neg$ by a fuzzy negation. We will present a brief study of some of the main properties that characterize this new class.  Finally, we will present two methods to obtain $t$-norms and $t$-conorms from an ($S$,$N$,$T$)-implication and a fuzzy negation.

This paper is organized as follows: the Section 2 contains some basic definitions and results that will be useful in subsequent sections;  in Section 3 we present the definition of ($S$,$N$,$T$)-implication, some of the main properties commonly associated with fuzzy implications are discussed for this new class of implications, in addition,  we present two methods to get $t$-norms and $t$-conorms from an ($S$,$N$,$T$)-implication and a fuzzy negation. And finally, in Section 4 are presented the conclusions and some proposals for future investigations.

\section{Preliminary Concepts}

In this section we will present some definitions and results already known in the literature, which will be useful for the development of this paper.

\begin{defi}\label{def2.1}(\cite{Baczynski2008,Klement2000})
A function $T:[0,1]^{2}\rightarrow[0,1]$ is said to be a \emph{triangular norm} (\emph{$t$-norm} for short) if it satisfies, for all $x,y,z,w\in[0,1]$, the following conditions:
\begin{enumerate}[label={\textbf{(T\arabic*)}}, ref=\textbf{(T\arabic*)}, align=left, leftmargin=*, noitemsep]
\item \label{itmT1} Commutativity: $T(x,y)=T(y,x)$;
\item \label{itmT2} Associativity: $T(x,T(y,z))=T(T(x,y),z)$;
\item \label{itmT3} Monotonicity: If $x\leq z$ and $y\leq w$ then $T(x,y)\leq T(z,w)$;
\item \label{itmT4} Boundary condition: $T(x,1)=x$.
\end{enumerate}
\end{defi}

\begin{obs}\label{rmk2.1}(\cite{Beliakov2007,Klement2000})
We can deduce from Definition \ref{def2.1} that each $t$-norm $T$ satisfies the following conditions:
\begin{enumerate}[label={\textbf{(T\arabic*)}}, ref=\textbf{(T\arabic*)}, align=left, leftmargin=*, noitemsep]
\setcounter{enumi}{4}
\item \label{itmT5} $T(0,x)=T(x,0)=0, \ \forall x\in[0,1]$;
\item \label{itmT6} $T(1,x)=x, \ \forall x\in[0,1]$;
\item \label{itmT7} $T(x,y)\leq x, \ \forall x,y\in[0,1]$;
\item \label{itmT8} $T(x,y)\leq y, \ \forall x,y\in[0,1]$.
\end{enumerate}
\end{obs}

\begin{exem}\label{ex2.1}
The four basic t-norms are $T_{min}(x,y)=\min\{x,y\}$ (\emph{minimum}), $T_{P}(x,y)=xy$ (\emph{product}), $T_{L}(x,y)=\max\{x+y-1,0\}$ (\emph{{\L}ukasiewicz $t$-norm}) and $T_{D}(x,y) = \begin{cases} 0 &\mbox{if } x,y\in[0,1[ \\
\min\{x,y\} & \mbox{otherwise. } \end{cases}$ (\emph{drastic product}).
\end{exem}

\begin{defi}\label{def2.2}(\cite{Baczynski2008,Beliakov2007})
A $t$-norm $T$ is called
\begin{enumerate}[label={\textbf{(\roman*)}}, ref=\textbf{(S\arabic*)}, align=left, leftmargin=*, noitemsep]
\item \emph{continuous} if it is continuous in both the arguments;
\item \emph{idempotent}, if $T(x,x)=x$ for all $x\in[0, 1]$;
\item \emph{strict}, if it is continuous and strictly monotone, i.e., $T(x,y)<T(x,z)$
whenever $x>0$ and $y<z$;
\item \emph{positive}, if $T(x,y)=0$ then either $x=0$ or $y=0$.
\end{enumerate}
\end{defi}

\begin{defi}\label{def2.3}(\cite{Baczynski2008,Klement2000})
A function $S:[0,1]^{2}\rightarrow[0,1]$ is said to be a \emph{triangular conorm} (\emph{t-conorm} for short) if it satisfies, for all $x,y,z,w\in[0,1]$, the following conditions:
\begin{enumerate}[label={\textbf{(S\arabic*)}}, ref=\textbf{(S\arabic*)}, align=left, leftmargin=*, noitemsep]
\item \label{itmS1} Commutativity: $S(x,y)=S(y,x)$;
\item \label{itmS2} Associativity: $S(x,S(y,z))=S(S(x,y),z)$;
\item \label{itmS3} Monotonicity: If $x\leq z$ and $y\leq w$ then $S(x,y)\leq S(z,w)$;
\item \label{itmS4} Boundary condition: $S(x,0)=x$.
\end{enumerate}
\end{defi}

\begin{obs}\label{rmk2.2}(\cite{Beliakov2007,Klement2000})
We can deduce from Definition \ref{def2.3} that each $t$-conorm $S$ satisfies the following conditions:
\begin{enumerate}[label={\textbf{(S\arabic*)}}, ref=\textbf{(S\arabic*)}, align=left, leftmargin=*, noitemsep]
\setcounter{enumi}{4}
\item \label{itmS5} $S(1,x)=S(x,1)=1, \ \forall x\in[0,1]$;
\item \label{itmS6} $S(0,x)=x, \ \forall x\in[0,1]$;
\item \label{itmS7} $S(x,y)\geq x, \ \forall x,y\in[0,1]$;
\item \label{itmS8} $S(x,y)\geq y, \ \forall x,y\in[0,1]$.
\end{enumerate}
\end{obs}

\begin{exem}\label{ex2.2}
The four basic $t$-conorms are $S_{max}(x,y)=\max\{x,y\}$ (\emph{maximum}), $S_{P}(x,y)=x+y-xy$ (\emph{probabilistic sum}), $S_{L}(x,y)=\min\{x+y,1\}$ (\emph{{\L}ukasiewicz $t$-conorm}) and $S_{D}(x,y)=\begin{cases} 1 &\mbox{if } x,y\in]0,1] \\
\max\{x,y\} & \mbox{otherwise. } \end{cases}$ (\emph{drastic sum}).
\end{exem}

\begin{defi}\label{def2.4}(\cite{Baczynski2008,Beliakov2007})
A $t$-conorm $S$ is called
\begin{enumerate}[label={\textbf{(\roman*)}}, ref=\textbf{(S\arabic*)}, align=left, leftmargin=*, noitemsep]
\item \emph{continuous} if it is continuous in both the arguments;
\item \emph{idempotent}, if $S(x,x)=x$ for all $x\in[0, 1]$;
\item \emph{strict}, if it is continuous and strictly monotone, i.e., $S(x,y)<S(x,z)$
whenever $x<1$ and $y<z$;
\item \emph{positive}, if $S(x,y)=1$ then either $x=1$ or $y=1$.
\end{enumerate}
\end{defi}

\begin{prop}\label{prop2.1}(\cite{Klement2000})
A function $S:[0,1]^{2}\rightarrow[0,1]$ is a t-conorm if and only if there exists a t-norm $T$ such that for all $x,y\in[0,1]$, $$S(x,y)=1-T(1-x,1-y).$$
\end{prop}

\begin{defi}\label{def2.5}(\cite{Baczynski2008})
A function $N:[0,1]\rightarrow[0,1]$ is said to be a \emph{fuzzy negation} if the following conditions hold:

\begin{enumerate}[label={\textbf{(N\arabic*)}}, ref=\textbf{(N\arabic*)}, align=left, leftmargin=*, noitemsep]

\item \label{itmN1} $N$ satisfies the boundary conditions: $N(0)=1$ and $N(1)=0$;

\item \label{itmN2} $N$ is non-increasing: if $x\leq y$ then $N(y)\leq N(x)$.
\end{enumerate}

Some extra properties for fuzzy negation are:


\begin{enumerate}[label={\textbf{(N\arabic*)}}, ref=\textbf{(N\arabic*)}, align=left, leftmargin=*, noitemsep]

\setcounter{enumi}{2}

\item \label{itmN3} $N$ is strictly decreasing;

\item \label{itmN4} $N$ is continuous;

%
%
\item \label{itmN5} $\forall x\in[0,1] : N(N(x))=x$;


\item \label{itmN6} $N(x)=1$ if and only if $x=0$;


\item \label{itmN7} $\forall x\in[0,1] : N(x)\in\{0,1\}$.
\end{enumerate}
A fuzzy negation $N$ is called of \textbf{strict}, \textbf{strong}, \textbf{non-filling} and \textbf{crisp} if satisfy \ref{itmN3} and  \ref{itmN4}, \ref{itmN5}, \ref{itmN6} and \ref{itmN7}, respectively.
\end{defi}

\begin{obs}\label{rmk2.10}(\cite{Dimuro2017})
A fuzzy negation $N:[0,1]\rightarrow[0,1]$ is crisp if and only if there exists $\alpha\in[0,1[$ such that $N=N_{\alpha}$ or there exists $\alpha\in]0,1]$ such that $N=N^{\alpha}$, where
\begin{equation}
N_{\alpha}(x) = \begin{cases} 0 &\mbox{if } x>\alpha \\
1 & \mbox{if } x\leq\alpha \end{cases} \quad \text{and} \quad N^{\alpha}(x) = \begin{cases} 0 &\mbox{if } x\geq\alpha \\
1 & \mbox{if } x<\alpha. \end{cases}
\end{equation}
\end{obs}

\begin{obs}
$N_{0}=N_{\bot}$ and $N^{1}=N_{\top}$.
\end{obs}

\begin{defi}\label{def2.6}(Definition 2.3.14 in \cite{Baczynski2008})
Let $T$ be a $t$-norm and $N$ a fuzzy negation. We say that the
pair ($T$,$N$) satisfies the \emph{law of contradiction} if and only if for all $x\in[0,1]$,
\begin{equation*}\label{LC}
T(N(x),x)=0 \tag{\textbf{LC}}.
\end{equation*}
\end{defi}

\begin{defi}\label{def2.7}(Definition 1.54 in \cite{Beliakov2007})
Let $T$ be a $t$-norm, $S$ a $t$-conorm and $N$ a strong fuzzy negation.  $T$ is \emph{$N$-dual to $S$} if and only if for all $x,y\in[0,1]$,
\begin{equation*}\label{T-(N-D)-S}
N(T(x,y))=S(N(x),N(y)) \tag{\textbf{T-(N-D)-S}}.
\end{equation*}
Analogously,  $S$ is \emph{$N$-dual to $T$} if and only if for all $x,y\in[0,1]$,
\begin{equation*}\label{S-(N-D)-T}
N(S(x,y))=T(N(x),N(y)) \tag{\textbf{S-(N-D)-T}}.
\end{equation*}
\end{defi}

\begin{defi}\label{def2.8}(\cite{Baczynski2008})
Let $T$ be a $t$-norm, $S$ a $t$-conorm and $N$ a strict negation. ($T$,$S$,$N$) is called a \emph{De Morgan triple} if they satisfy the following equations:
\begin{enumerate}[label={\textbf{(\roman*)}}, ref=\textbf{(\roman*)}, align=left, leftmargin=*, noitemsep]
\item $N(S(x,y))=T(N(x),N(y)), \forall x,y\in[0,1]$;
\item $N(T(x,y))=S(N(x),N(y)), \forall x,y\in[0,1]$.
\end{enumerate}
\end{defi}

\begin{teo}\label{thm2.1}(\cite{Baczynski2008})
Let $T$ be a $t$-norm, $S$ a $t$-conorm and $N$ a strict negation.  Then, ($T$,$S$,$N$) is a De Morgan triple if and only if $N$ is a strong negation and $S$ is $N$-dual to $T$.
\end{teo}

\begin{defi}\label{def2.9}(\cite{Baczynski2013,Baczynski2008})
A function $I:[0,1]^{2}\rightarrow[0,1]$ is said to be a \emph{fuzzy implication} if it satisfies the following conditions:
\begin{enumerate}[label={\textbf{(I\arabic*)}}, ref=\textbf{(I\arabic*)}, align=left, leftmargin=*, noitemsep]
\item \label{itmI1} $I(x,z)\geq I(y,z)$ whenever $x\leq y$ and $z\in[0,1]$; 
\item \label{itmI2} $I(x,y)\leq I(x,z)$ whenever $y\leq z$ and $x\in[0,1]$; 
\item \label{itmI3} $I(0,0)=1$;
\item \label{itmI4} $I(1,1)=1$;
\item \label{itmI5} $I(1,0)=0$.
\end{enumerate}
\end{defi}

We can deduce directly from Definition \ref{def2.9} that each fuzzy implication $I$ satisfies the following properties:
\begin{enumerate}[label={\textbf{(I\arabic*)}}, ref=\textbf{(I\arabic*)}, align=left, leftmargin=*, noitemsep]
\setcounter{enumi}{5}
\item \label{itmI6} $I(0,y)=1, \ \forall y\in[0,1]$. (left boundary condition)
\item \label{itmI7} $I(x,1)=1, \ \forall x\in[0,1]$. (right boundary condition)
\item \label{itmI8} $I(x,y)\geq I(x,0), \ \forall x,y\in[0,1]$.
\item \label{itmI9} $I(x,y)\geq I(1,y), \ \forall x,y\in[0,1]$.
\end{enumerate}

The set of all fuzzy implications will be denoted by $\mathcal{FI}$. There exist several properties that may be required for fuzzy implications \cite{Baczynski2008,Dimuro2015,Dimuro2019,Dimuro2014}. In what follows, we present some of these properties, which will be used to characterize the class of fuzzy implications proposed in this paper.

\begin{defi}\label{def2.10}
A fuzzy implication $I:[0,1]^{2}\rightarrow[0,1]$ satisfies:

\begin{description}
\descitem{(NP)} if $\forall y\in[0,1] : I(1,y)=y$;
\descitem{(EP)} if $\forall x,y,z\in[0,1] : I(x,I(y,z))=I(y,I(x,z))$;
\descitem{(IP)} if  $\forall x\in[0,1] : I(x,x)=1$;
\descitem{(LOP)} if $\forall x,y\in[0,1] : x\leq y \implies I(x,y)=1$;
\descitem{(ROP)} if $\forall x,y\in[0,1] : I(x,y)=1 \implies x\leq y$;
\descitem{(CB)} if $\forall x,y\in[0,1] : y\leq I(x,y)$;
\descitem{(SIB)} if $\forall x,y\in[0,1] : I(x,I(x,y))\geq I(x,y)$;
\descitem{(IB)} if $\forall x,y\in[0,1] : I(x,I(x,y))=I(x,y)$.
\end{description}
\end{defi}

\begin{defi}\label{def2.11}(\cite{Baczynski2008})
Let $I$ be a fuzzy implication and $N$ be a fuzzy negation. $I$ satisfies
\begin{description}
\descitem{(CP)} if $\forall x,y\in[0,1] : I(x,y)=I(N(y),N(x))$;
\descitem{(L-CP)} if $\forall x,y\in[0,1] : I(N(x),y)=I(N(y),x)$;
\descitem{(R-CP)} if $\forall x,y\in[0,1] : I(x,N(y))=I(y,N(x))$.
\end{description}
\end{defi}

If $I$ satisfies \textbf{\descref{(CP)}} (\textbf{\descref{(L-CP)}}, \textbf{\descref{(R-CP)}}) with respect to a specific $N$, then we will denote this by \textbf{\descref{(CP)}}($N$) (respectively, by \textbf{\descref{(L-CP)}}($N$), \textbf{\descref{(R-CP)}}($N$)).

\begin{defi}\label{def2.12}(\cite{Baczynski2008,Baczynski2013})
Let $I\in\mathcal{FI}$. The function $N_{I}$ defined by 
\begin{equation}\label{eq1}
 N_{I}(x):=I(x,0),
\end{equation}
is called the
\emph{natural negation of $I$} or the \emph{negation induced by $I$}.
\end{defi}

\begin{defi}\label{def2.13}(\cite{Baczynski2007,Baczynski2008})
A function $I:[0,1]^{2}\rightarrow[0,1]$ is called an ($S$,$N$)-implication if there exist a $t$-conorm $S$ and a fuzzy negation $N$ such that
\begin{equation}\label{eq2}
I(x,y)=S(N(x),y)
\end{equation}
for all $x,y\in[0,1]$.  We will write $I_{S,N}$ to denote an ($S$,$N$)-implication.
\end{defi}

\begin{prop}\label{prop2.2}(\cite{Baczynski2007,Baczynski2008})
If $I_{S,N}$ is an ($S$,$N$)-implication, then
\begin{enumerate}[label={\textbf{(\roman*)}}, ref=\textbf{(\roman*)}, align=left, leftmargin=*, noitemsep]
\item $I_{S,N}\in\mathcal{FI}$ and $I_{S,N}$ satisfies \textbf{\descref{(NP)}} and \textbf{\descref{(EP)}};
\item $N_{I_{S,N}}=N$;
\item $I_{S,N}$ satisfies \textbf{\descref{(R-CP)}}($N$);
\item If $N$ is strict then $I_{S,N}$ satisfies \textbf{\descref{(L-CP)}}($N^{-1}$);
\item If $N$ is strong then $I_{S,N}$ satisfies \textbf{\descref{(CP)}}($N$).
\end{enumerate}
\end{prop}

\begin{defi}\label{def2.14}(\cite{Bedregal2007,Pinheiro2017,Pinheiro2018})
A function $I:[0,1]^{2}\rightarrow[0,1]$ is called an ($T$,$N$)-implication if there exist a $t$-norm $T$ and a fuzzy negation $N$ such that
\begin{equation}\label{eq3}
I(x,y)=N(T(x,N(y)))
\end{equation}
for all $x,y\in[0,1]$. We will write $I_{T,N}$ to denote an ($T$,$N$)-implication.
\end{defi}

\begin{prop}\label{prop2.3}(\cite{Bedregal2007,Pinheiro2017,Pinheiro2018})
$I_{T,N}\in\mathcal{FI}$.
\end{prop}

\begin{prop}\label{prop2.4}(\cite{Bedregal2007,Pinheiro2017,Pinheiro2018})
Let $N$ be a strong fuzzy negation and let $T$ be a t-norm. Then, for all $x,y\in[0,1]$,
$$T(x,y)=N(I_{T,N}(x,N(y))).$$
\end{prop}

For more results on ($S$,$N$) and ($T$,$N$)-implications see \cite{Baczynski2007,Baczynski2008a,Pinheiro2017,Pinheiro2018}.

\section{($S$,$N$,$T$)-Implications}

In this section we will introduce a new class of fuzzy implications obtained from the composition of a $t$-conorm, a fuzzy negation and a $t$-norm.

\begin{defi}\label{def3.1}
A function $I:[0,1]^{2}\rightarrow[0,1]$ is called an \emph{($S$,$N$,$T$)-implication} if there exist a $t$-conorm $S$, a fuzzy negation $N$ and a $t$-norm $T$ such that
\begin{equation}\label{eq4}
    I(x,y) = S(N(T(x,N(y))),N(x))
\end{equation}
for all $x,y\in[0,1]$. We will write $I_{S,N,T}$ to denote an \emph{($S$,$N$,$T$)-implication}.
\end{defi}

\begin{prop}\label{prop3.1}
$I_{S,N,T}\in\mathcal{FI}$.
\end{prop}

\begin{proof}
It follows that:\vspace{0.2cm}
\\
\textbf{(i)} $I_{S,N,T}$ satisfies \ref{itmI3}, \ref{itmI4} and \ref{itmI5}. In fact,
\begin{eqnarray*}
I_{S,N,T}(0,0) &=& S(N(T(0,1)),1)\overset{\ref{itmT4}}{=}S(N(0),1)=S(1,1)\overset{\ref{itmS5}}{=}1\\
I_{S,N,T}(1,1) &=& S(N(T(1,0)),0)\overset{\ref{itmT5}}{=}S(N(0),0)=S(1,0)\overset{\ref{itmS4}}{=}1 \\
I_{S,N,T}(1,0) &=& S(N(T(1,1)),0)\overset{\ref{itmT4}}{=}S(N(1),0)=S(0,0)\overset{\ref{itmS4}}{=}0
\end{eqnarray*}
\textbf{(ii)} $I_{S,N,T}$ satisfies \ref{itmI1}. In fact, let $x,y,z\in[0,1]$ such that $x\leq y$. Hence,
\begin{eqnarray*}
x\leq y &\overset{\ref{itmT3},\ref{itmN2}}{\implies}&
N(T(y,N(z)))\leq N(T(x,N(z)))\\
&\overset{\ref{itmN2},\ref{itmS3}}{\implies}& S(N(T(y,N(z))),N(y))\leq S(N(T(x,N(z))),N(x))\\
&\overset{\eqref{eq4}}{\implies}& I_{S,N,T}(y,z)\leq I_{S,N,T}(x,z).
\end{eqnarray*}
\textbf{(iii)} $I_{S,N,T}$ satisfies \ref{itmI2}. In fact, let $x,y,z\in[0,1]$ such that $y\leq z$. Hence,
\begin{eqnarray*}
y\leq z &\overset{\ref{itmN2},\ref{itmT3}}{\implies}&
N(T(x,N(y)))\leq N(T(x,N(z)))\\
&\overset{\ref{itmS3}}{\implies}& S(N(T(x,N(y))),N(x))\leq S(N(T(x,N(z))),N(x))\\
&\overset{\eqref{eq4}}{\implies}& I_{S,N,T}(x,y)\leq I_{S,N,T}(x,z).
\end{eqnarray*}
Therefore, $I_{S,N,T}$ is a fuzzy implication.
\end{proof}

\begin{prop}\label{prop3.2}
Let $I_{S,N,T}$ be an ($S$,$N$,$T$)-implication. Then it holds that:
\begin{enumerate}[label={\textbf{(\roman*)}}, ref=\textbf{(\roman*)}, align=left, leftmargin=*, noitemsep]
\item $I_{S,N,T}$ satisfies \textbf{\descref{(NP)}} if and only if $N$ is strong.
\item If $N$ is strong then $I_{S,N,T}$ satisfies \textbf{\descref{(CB)}}.
\item If $N$ is strong then $I_{S,N,T}$ satisfies \textbf{\descref{(SIB)}}.
\item If $S=S_{max}$ or $N$ is crisp then $N_{I_{S,N,T}}=N$.
\item If $N$ is crisp then $I_{S,N,T}$ satisfies \textbf{\descref{(IP)}}.
\item If $N$ is crisp then $I_{S,N,T}$ satisfies \textbf{\descref{(LOP)}}.
\end{enumerate}
\end{prop}

\begin{proof}
It follows that:\vspace{0.2cm}
\\
\textbf{(i)}
\vspace{-0.7cm}
\begin{equation*}
\begin{split}
I_{S,N,T} \ \text{satisfies \textbf{\descref{(NP)}}} &\iff I_{S,N,T}(1,y)=y, \forall y\in[0,1]\\
&\overset{\ref{itmT6}}{\iff} S(N(N(y)),0)=y, \forall y\in[0,1]\\
&\overset{\ref{itmS4}}{\iff}
N \ \text{is strong}.
\end{split}
\end{equation*}
\textbf{(ii)} By \textbf{(i)}, one has that $I_{S,N,T}$ satisfies \textbf{\descref{(NP)}}, i.e. $I_{S,N,T}(1,y)=y$ for all $y\in[0,1]$. Hence, it follows that $I_{S,N,T}(x,y)\overset{\ref{itmI9}}{\geq} I_{S,N,T}(1,y)=y$ for all $x,y\in[0,1]$.
\vspace{0.3cm}
\\
\textbf{(iii)} By \textbf{(ii)}, one has that $y\leq I_{S,N,T}(x,y), \forall x,y\in[0,1]$. Hence, it follows that, $I_{S,N,T}(x,y)\overset{\ref{itmI2}}{\leq} I_{S,N,T}(x,I_{S,N,T}(x,y)), \forall x,y\in[0,1]$.
\vspace{0.3cm}
\\
\textbf{(iv)} For all $x\in[0,1]$, one has that
\begin{equation}\label{eq5}
N_{I_{S,N,T}}(x)=I_{S,N,T}(x,0)=S(N(T(x,1)),N(x))\overset{\ref{itmT4}}{=}S(N(x),N(x)).
\end{equation}
If $S=S_{max}$ then by \eqref{eq5} we get $N_{I_{S,N,T}}(x)=N(x)$ for all $x\in[0,1]$. Now, let $x\in[0,1]$. Since $N$ is crisp then $N(x)=0$ or $N(x)=1$. If $N(x)=0$ then by \eqref{eq5} we get $N_{I_{S,N,T}}(x)=S(N(x),N(x))=S(0,0)\overset{\ref{itmS4}}{=}0=N(x)$. If $N(x)=1$ then by \eqref{eq5} we get $N_{I_{S,N,T}}(x)=S(N(x),N(x))=S(1,1)\overset{\ref{itmS5}}{=}1=N(x)$.\\
\vspace{-0.1cm}
\\
\textbf{(v)} For all $x\in[0,1]$ one has that $I_{S,N,T}(x,x)=S(N(T(x,N(x))),N(x))$. Hence, if $N(x)=0$ then $I_{S,N,T}(x,x)=S(N(T(x,0)),0)\overset{\ref{itmT5}}{=}S(N(0),0)=S(1,0)\overset{\ref{itmS4}}{=}1$ and if $N(x)=1$ then $I_{S,N,T}(x,x)=S(N(T(x,1)),1)\overset{\ref{itmS5}}{=}1$. Therefore, $I_{S,N,T}$ satisfies \textbf{\descref{(IP)}}.
\vspace{0.3cm}
\\
\textbf{(vi)} Let $x,y\in[0,1]$ such that $x\leq y$. We have two possible cases:
\vspace{0.2cm}
\\
$N(x)=1$: In this case, it is obvious by \ref{itmS5} that $I_{S,N,T}(x,y)=1$.
\vspace{0.2cm}
\\
$N(x)=0$: By \ref{itmN2} one has that $N(y)\leq N(x)=0$, i.e., $N(y)=0$. Hence, it follows that $I_{S,N,T}(x,y)=S(N(T(x,0)),0)\overset{\ref{itmT5}}{=}S(N(0),0)=S(1,0)\overset{\ref{itmS4}}{=}1.$
Therefore, $I_{S,N,T}$ satisfies \textbf{\descref{(LOP)}}.
\end{proof}

The following result characterizes the ($S$,$N$,$T$)-implications when $N$ is a crisp fuzzy negation.

\begin{prop}
Let $T$ be a $t$-norm, $S$ be a $t$-conorm and $N$ be a crisp fuzzy negation. Then, it holds that:
\begin{enumerate}[label={\textbf{(\roman*)}}, ref=\textbf{(\roman*)}, align=left, leftmargin=*, noitemsep]
\item $I_{S,N,T}$ satisfies \textbf{\descref{(EP)}};
\item $I_{S,N,T}$ does not satisfy \textbf{\descref{(NP)}};
\item $I_{S,N,T}$ satisfies \textbf{\descref{(IP)}};
\item $I_{S,N,T}$ satisfies \textbf{\descref{(LOP)}};
\item $I_{S,N,T}$ does not satisfy \textbf{\descref{(ROP)}};
\item $I_{S,N,T}$ does not satisfy \textbf{\descref{(OP)}}.
\end{enumerate}
\end{prop}

\begin{proof}
Consider $N=N_{\alpha}$, for some $\alpha\in[0,1[$.
\vspace{0.2cm}
\\
\textbf{(\emph{i})} Let $x,y,z\in[0,1]$. We can divide the discussion into three possible cases.
\vspace{0.2cm}
\\
\textbf{Case 1}: $x\leq\alpha$.
\begin{eqnarray*}
I_{S,N_{\alpha},T}(x,I_{S,N_{\alpha},T}(y,z)) &=& S(N_{\alpha}(T(x,N_{\alpha}(S(N_{\alpha}(T(y,N_{\alpha}(z))),N_{\alpha}(y))))),N_{\alpha}(x)) \\
&=& S(N_{\alpha}(T(x,N_{\alpha}(S(N_{\alpha}(T(y,N_{\alpha}(z))),N_{\alpha}(y))))),1) \\
&\overset{\ref{itmS5}}{=}& 1
\end{eqnarray*}
and
\begin{eqnarray*}
I_{S,N_{\alpha},T}(y,I_{S,N_{\alpha},T}(x,z)) &=& S(N_{\alpha}(T(y,N_{\alpha}(S(N_{\alpha}(T(x,N_{\alpha}(z))),N_{\alpha}(x))))),N_{\alpha}(y)) \\
&=& S(N_{\alpha}(T(x,N_{\alpha}(S(N_{\alpha}(T(y,N_{\alpha}(z))),1)))),N_{\alpha}(y)) \\
&\overset{\ref{itmS5}}{=}& S(N_{\alpha}(T(x,N_{\alpha}(1))),N_{\alpha}(y)) \\
&=& S(N_{\alpha}(T(x,0)),N_{\alpha}(y)) \\
&\overset{\ref{itmT5}}{=}& S(N_{\alpha}(0),N_{\alpha}(y)) \\
&=& S(1,N_{\alpha}(y)) \\
&\overset{\ref{itmS5}}{=}& 1
\end{eqnarray*}
\textbf{Case 2}: $x>\alpha$ and $y\leq\alpha$.
\begin{eqnarray*}
I_{S,N_{\alpha},T}(x,I_{S,N_{\alpha},T}(y,z)) &=& S(N_{\alpha}(T(x,N_{\alpha}(S(N_{\alpha}(T(y,N_{\alpha}(z))),N_{\alpha}(y))))),N_{\alpha}(x)) \\
&=& S(N_{\alpha}(T(x,N_{\alpha}(S(N_{\alpha}(T(y,N_{\alpha}(z))),1)))),0) \\
&\overset{\ref{itmS5}}{=}& S(N_{\alpha}(T(x,N_{\alpha}(1))),0) \\
&=& S(N_{\alpha}(T(x,0)),0) \\
&\overset{\ref{itmT5}}{=}& S(N_{\alpha}(0),0) \\
&=& S(1,0) \\
&\overset{\ref{itmS4}}{=}& 1
\end{eqnarray*}
and
\begin{eqnarray*}
I_{S,N_{\alpha},T}(y,I_{S,N_{\alpha},T}(x,z)) &=& S(N_{\alpha}(T(y,N_{\alpha}(S(N_{\alpha}(T(x,N_{\alpha}(z))),N_{\alpha}(x))))),N_{\alpha}(y)) \\
&=& S(N_{\alpha}(T(y,N_{\alpha}(S(N_{\alpha}(T(x,N_{\alpha}(z))),0)))),1) \\
&\overset{\ref{itmS5}}{=}& 1
\end{eqnarray*}
\textbf{Case 3}: $x>\alpha$ and $y>\alpha$.

\textbf{Subcase 3.1}: $z\leq\alpha$.
\begin{eqnarray*}
I_{S,N_{\alpha},T}(x,I_{S,N_{\alpha},T}(y,z)) &=& S(N_{\alpha}(T(x,N_{\alpha}(S(N_{\alpha}(T(y,N_{\alpha}(z))),N_{\alpha}(y))))),N_{\alpha}(x)) \\
&=& S(N_{\alpha}(T(x,N_{\alpha}(S(N_{\alpha}(T(y,1)),0)))),0) \\
&\overset{\ref{itmT4}}{=}& S(N_{\alpha}(T(x,N_{\alpha}(S(N_{\alpha}(y),0)))),0) \\
&=& S(N_{\alpha}(T(x,N_{\alpha}(S(0,0)))),0) \\
&\overset{\ref{itmS4}}{=}& S(N_{\alpha}(T(x,N_{\alpha}(0))),0) \\
&=& S(N_{\alpha}(T(x,1)),0) \\
&\overset{\ref{itmT4}}{=}& S(N_{\alpha}(x),0) \\
&=& S(0,0) \\
&\overset{\ref{itmS4}}{=}& 0
\end{eqnarray*}

and
\begin{eqnarray*}
I_{S,N_{\alpha},T}(y,I_{S,N_{\alpha},T}(x,z)) &=& S(N_{\alpha}(T(y,N_{\alpha}(S(N_{\alpha}(T(x,N_{\alpha}(z))),N_{\alpha}(x))))),N_{\alpha}(y)) \\
&=& S(N_{\alpha}(T(y,N_{\alpha}(S(N_{\alpha}(T(x,1)),0)))),0) \\
&\overset{\ref{itmT4}}{=}& S(N_{\alpha}(T(y,N_{\alpha}(S(N_{\alpha}(x),0)))),0) \\
&=& S(N_{\alpha}(T(y,N_{\alpha}(S(0,0)))),0) \\
&\overset{\ref{itmS4}}{=}& S(N_{\alpha}(T(y,N_{\alpha}(0))),0) \\
&=& S(N_{\alpha}(T(y,1)),0) \\
&\overset{\ref{itmT4}}{=}& S(N_{\alpha}(y),0) \\
&=& S(0,0) \\
&\overset{\ref{itmS4}}{=}& 0
\end{eqnarray*}

\textbf{Subcase 3.2}: $z>\alpha$.
\begin{eqnarray*}
I_{S,N_{\alpha},T}(x,I_{S,N_{\alpha},T}(y,z)) &=& S(N_{\alpha}(T(x,N_{\alpha}(S(N_{\alpha}(T(y,N_{\alpha}(z))),N_{\alpha}(y))))),N_{\alpha}(x)) \\
&=& S(N_{\alpha}(T(x,N_{\alpha}(S(N_{\alpha}(T(y,0)),0)))),0) \\
&\overset{\ref{itmT5}}{=}& S(N_{\alpha}(T(x,N_{\alpha}(S(N_{\alpha}(0),0)))),0) \\
&=& S(N_{\alpha}(T(x,N_{\alpha}(S(1,0)))),0) \\
&\overset{\ref{itmS4}}{=}& S(N_{\alpha}(T(x,N_{\alpha}(1))),0) \\
&=& S(N_{\alpha}(T(x,0)),0) \\
&\overset{\ref{itmT5}}{=}& S(N_{\alpha}(0),0) \\
&=& S(1,0) \\
&\overset{\ref{itmS4}}{=}& 1
\end{eqnarray*}

and
\begin{eqnarray*}
I_{S,N_{\alpha},T}(y,I_{S,N_{\alpha},T}(x,z)) &=& S(N_{\alpha}(T(y,N_{\alpha}(S(N_{\alpha}(T(x,N_{\alpha}(z))),N_{\alpha}(x))))),N_{\alpha}(y)) \\
&=& S(N_{\alpha}(T(y,N_{\alpha}(S(N_{\alpha}(T(x,0)),0)))),0) \\
&\overset{\ref{itmT5}}{=}& S(N_{\alpha}(T(y,N_{\alpha}(S(N_{\alpha}(0),0)))),0) \\
&=& S(N_{\alpha}(T(y,N_{\alpha}(S(1,0)))),0) \\
&\overset{\ref{itmS4}}{=}& S(N_{\alpha}(T(y,N_{\alpha}(1))),0) \\
&=& S(N_{\alpha}(T(y,0)),0) \\
&\overset{\ref{itmT5}}{=}& S(N_{\alpha}(0),0) \\
&=& S(1,0) \\
&\overset{\ref{itmS4}}{=}& 1
\end{eqnarray*}
In any case, one has that $I_{S,N_{\alpha},T}(x,I_{S,N_{\alpha},T}(y,z))=I_{S,N_{\alpha},T}(y,I_{S,N_{\alpha},T}(x,z))$ for all $x,y,z\in[0,1]$. Therefore, $I_{S,N_{\alpha},T}$ satisfies \textbf{\descref{(EP)}}.
\vspace{0.2cm}
\\
\textbf{(\emph{ii})} As $N$ is crisp then $N$ is not strong. Hence, by Proposition \ref{prop3.2}(\emph{i}), we conclude that $I_{S,N,T}$ does not satisfy \textbf{\descref{(NP)}}.
\vspace{0.2cm}
\\
\textbf{(\emph{iii})} Let $x\in[0,1]$. If $x\leq\alpha$ then, by \ref{itmS5}, $I_{S,N_{\alpha},T}(x,x)=1$. If $x>\alpha$ then
\begin{eqnarray*}
I_{S,N_{\alpha},T}(x,x) &=& S(N_{\alpha}(T(x,N_{\alpha}(x))),N_{\alpha}(x))=S(N_{\alpha}(T(x,0)),0) \\
&\overset{\ref{itmT5}}{=}& S(N_{\alpha}(0),0)=S(1,0)\overset{\ref{itmS4}}{=}1.
\end{eqnarray*}
Therefore, $I_{S,N_{\alpha},T}$ satisfies \textbf{\descref{(IP)}}.
\vspace{0.2cm}
\\
\textbf{(\emph{iv})} Let $x,y\in[0,1]$ such that $x\leq y$. If $x=y$ then, by (\emph{\textbf{vi}}), one has that $I_{S,N_{\alpha},T}(x,y)=1$. Now, suppose that $x<y$. There are three possible cases to consider.
\vspace{0.2cm}
\\
$\alpha<x<y$: In this case,
$$I_{S,N_{\alpha},T}(x,y)=S(N_{\alpha}(T(x,N_{\alpha}(y))),N_{\alpha}(x))=S(N_{\alpha}(T(x,0)),0)=1.$$
$x\leq\alpha<y$: In this case, we obtain immediately by \ref{itmS5} that $I_{S,N_{\alpha},T}(x,y)=1$.
\vspace{0.2cm}
\\
$x<y\leq\alpha$: In this case, we obtain immediately by \ref{itmS5} that $I_{S,N_{\alpha},T}(x,y)=1$.
\vspace{0.2cm}
\\
Therefore, $I_{S,N_{\alpha},T}$ satisfies \textbf{\descref{(LOP)}}.
\vspace{0.2cm}
\\
\textbf{(\emph{v})} Let $x,y\in[0,1]$ such that $x>y\geq\alpha$. So, $N_{\alpha}(x)=1$. Hence, it follows immediately by \ref{itmS5} that $I_{S,N_{\alpha},T}(x,y)=1$. Therefore, $I_{S,N_{\alpha},T}$ does not satisfy \textbf{\descref{(ROP)}}.
\vspace{0.2cm}
\\
\textbf{(\emph{vi})} The proof follows immediately from item (\emph{\textbf{v}}).
\vspace{0.2cm}
\\
Analogously we can prove each of the above items for $N=N^{\alpha}$.
\end{proof}

The next result establishes some relations between ($S$,$N$), ($T$,$N$) and ($S$,$N$,$T$)-implications for a given triple ($T$,$S$,$N$).

\begin{prop}\label{prop3.3}
Let $S$ be a $t$-conorm, $N$ a fuzzy negation and $T$ a $t$-norm. Then it holds that:
\begin{enumerate}[label={\textbf{(\roman*)}}, ref=\textbf{(\roman*)}, align=left, leftmargin=*, noitemsep]
\item $I_{S,N,T}\geq I_{T,N}$;
\item If $N$ is strong then $I_{S,N,T}(x,y)=I_{S,N}(x,I_{T,N}(x,y))$ for all $x,y\in[0,1]$;
\item If ($T$,$S$,$N$) is a De Morgan triple then \vspace{0.1cm}
\\
\hspace{-3.5cm}$(a)$ $I_{S,N,T}(x,y)=I_{S,N}(x,I_{S,N}(x,y))$ for any $x,y\in[0,1]$;\vspace{0.1cm}
\\
\hspace{-3.5cm}$(b)$ $I_{S,N,T}(x,y)=I_{T,N}(x,I_{T,N}(x,y))$ for any $x,y\in[0,1]$.
\end{enumerate}
\end{prop}

\begin{proof}
\textbf{(i)} $I_{S,N,T}(x,y)=S(I_{T,N}(x,y),N(x))\overset{\ref{itmS7}}{\geq}I_{T,N}(x,y)$ for all $x,y\in[0,1]$.
\textbf{(ii)} Since $N$ is strong then $T(x,N(y))=N(I_{T,N}(x,y)), \forall x,y\in[0,1]$ (Proposition \ref{prop2.4}) and $I_{S,N}$ satisfies \textbf{\descref{(CP)}}($N$) (Proposition \ref{prop2.2}\textbf{$(v)$}), i.e. $I_{S,N}(x,y)=I_{S,N}(N(y),N(x))$ for all $x,y\in[0,1]$. Hence, it follows that
\begin{equation*}
\begin{split}
I_{S,N,T}(x,y) &= S(N(T(x,N(y))),N(x)) =I_{S,N}(T(x,N(y)),N(x)) \\
&= I_{S,N}(N(I_{T,N}(x,y)),N(x))= I_{S,N}(x,I_{T,N}(x,y)).
\end{split}
\end{equation*}
\textbf{(iii)}$(a)$ From Theorem \ref{thm2.1}, $N$ is strong and $S$ is $N$-dual to $T$. Thereby,
\begin{eqnarray*}
I_{S,N,T}(x,y) &=& S(N(T(x,N(y))),N(x))\\
&\overset{\ref{itmN5}}{=}& S(N(T(N(N(x)),N(y))),N(x))\\
&\overset{\ref{S-(N-D)-T}}{=}& S(N(N(S(N(x),y)),N(x))\\
&\overset{\ref{itmN5},\ref{itmS1}}{=}& S(N(x),S(N(x),y))
= I_{S,N}(x,I_{S,N}(x,y)).
\end{eqnarray*}
\textbf{(iii)}$(b)$ The proof is analogous.
\end{proof}

\begin{prop}
Let $\alpha\in[0,1[$. Then, $I_{S,N_{\alpha},T}$ is the fuzzy implication defined by
$$I_{S,N_{\alpha},T}(x,y) = \begin{cases} 1 &\mbox{if } x\leq\alpha \ \text{or} \ y>\alpha \\
0 & \mbox{otherwise. } \end{cases}$$
\end{prop}

\begin{proof}
Just calculate the value of $I_{S,N_{\alpha},T}$ in the following cases:
\vspace{0.2cm}
\\
$x\leq\alpha \ \text{or} \ y>\alpha$: From the definition of $I_{S,N_{\alpha},T}$ and $N_{\alpha}$, we get that
\begin{eqnarray*}
I_{S,N_{\alpha},T}(x,y) &=& S(N_{\alpha}(T(x,N_{\alpha}(y))),N_{\alpha}(x)) = S(N_{\alpha}(T(x,0)),1) \\
&\overset{\ref{itmT5}}{=}& S(N_{\alpha}(0),1)=S(1,1)\overset{\ref{itmS5}}{=}1
\end{eqnarray*}
$x>\alpha \ \text{and} \ y\leq\alpha$: From the definition of $I_{S,N_{\alpha},T}$ and $N_{\alpha}$, we get that
\begin{eqnarray*}
I_{S,N_{\alpha},T}(x,y) &=& S(N_{\alpha}(T(x,N_{\alpha}(y))),N_{\alpha}(x)) = S(N_{\alpha}(T(x,1)),0) \\
&\overset{\ref{itmT4}}{=}& S(N_{\alpha}(x),0) = S(0,0)\overset{\ref{itmS4}}{=}0,
\end{eqnarray*}
which concludes the proof of the result.
\end{proof}

Analogously we can prove the following

\begin{prop}
Let $\alpha\in]0,1]$. Then, $I_{S,N^{\alpha},T}$ is the fuzzy implication defined by
$$I_{S,N^{\alpha},T}(x,y) = \begin{cases} 1 &\mbox{if } x<\alpha \ \text{or} \ y\geq\alpha \\
0 & \mbox{otherwise. } \end{cases}$$
\end{prop}

The next result presents a characterization of ($S$,$N$,$T$)-implications through the law of contradiction.

\begin{teo}\label{thm3.1}
Let $T$ be a $t$-norm, $S$ a positive $t$-conorm and $N$ a non-filling fuzzy negation. Then, $I_{S,N,T}$ satisfies \textbf{\descref{(IP)}} if and only if ($T$,$N$) satisfies \eqref{LC}.
\end{teo}

\begin{proof}
($\Leftarrow$) We have that ($T$,$N$) satisfies \eqref{LC}. 
So, it follows that
\begin{eqnarray*}
I_{S,N,T}(x,x) &=& S(N(T(x,N(x))),N(x)) \overset{\ref{itmT1}}{=} S(N(T(N(x),x)),N(x)) \\
&\overset{\eqref{LC}}{=}& S(N(0),N(x)) = S(1,N(x)) \overset{\ref{itmS5}}{=} 1.
\end{eqnarray*}
($\Rightarrow$) Suppose that ($T$,$N$) does not satisfy \eqref{LC}. Then there exists $x_{0}\in[0,1]$ such that $T(N(x_{0}),x_{0})\neq0$. From \ref{itmT5} one has that $x_{0}\neq0$ and $N(x_{0})\neq0$. As $N$ is non-filling (\ref{itmN6}) then $N(x_{0})<1$ and $N(T(N(x_{0}),x_{0}))<1$. Hence, as $S$ is positive then it follows that $I_{S,N,T}(x_{0},x_{0})=S(N(T(x_{0},N(x_{0}))),N(x_{0}))<1$,
i.e., $I_{S,N,T}$ does not satisfy \textbf{\descref{(IP)}}.
\end{proof}

\begin{coro}\label{cor3.1}
Let $T$ be a $t$-norm, $S$ a strict $t$-conorm and $N$ a non-filling fuzzy negation. Then, $I_{S,N,T}$ satisfies \textbf{\descref{(IP)}} if and only if ($T$,$N$) satisfies \eqref{LC}.
\end{coro}

\begin{proof}
It follows directly from the fact that all strict $t$-conorm is positive and from Theorem \ref{thm3.1}.
\end{proof}

\begin{prop}\label{prop3.4}
Let $S$ be a $t$-conorm, $T$ a $t$-norm and $N$ a strong fuzzy negation. If $T$ is $N$-dual to $S$ then $I_{S,N,T}$ satisfies \textbf{\descref{(EP)}}.
\end{prop}

\begin{proof}
It follows that:
\begin{eqnarray}\label{eq6}
I_{S,N,T}(x,I_{S,N,T}(y,z)) &=& S(N(T(x,N(S(N(T(y,N(z))),N(y))))),N(x)) \nonumber \\ &\overset{\ref{T-(N-D)-S}}{=}& S(N(T(x,N(N(T(T(y,N(z)),y))))),N(x)) \nonumber \\ 
&\underset{\ref{itmT1},\ref{itmT2}}{\overset{\ref{itmN5}}{=}}& S(N(T(x,T(T(y,y),N(z)))),N(x)) \nonumber \\
&\underset{\ref{T-(N-D)-S}}{=}& N(T(T(x,T(T(y,y),N(z))),x)) \nonumber \\
&\underset{\ref{itmT1},\ref{itmT2}}{=}& 
N(T(T(T(x,x),T(y,y)),N(z)))
\end{eqnarray}
and
\begin{eqnarray}\label{eq7}
I_{S,N,T}(y,I_{S,N,T}(x,z)) &=& S(N(T(y,N(S(N(T(x,N(z))),N(x))))),N(y)) \nonumber \\ &\overset{\ref{T-(N-D)-S}}{=}& S(N(T(y,N(N(T(T(x,N(z)),x))))),N(y)) \nonumber \\ &\underset{\ref{itmT1},\ref{itmT2}}{\overset{\ref{itmN5}}{=}}& S(N(T(y,T(T(x,x),N(z)))),N(y)) \nonumber \\
&\underset{\ref{T-(N-D)-S}}{=}& N(T(T(y,T(T(x,x),N(z))),y)) \nonumber \\
&\underset{\ref{itmT1},\ref{itmT2}}{=}&
N(T(T(T(x,x),T(y,y)),N(z)))
\end{eqnarray}
\vspace{-0.3cm}
\\
Thus, by \eqref{eq6} and \eqref{eq7} we get $I_{S,N,T}(x,I_{S,N,T}(y,z))=I_{S,N,T}(y,I_{S,N,T}(x,z))$ for all $x,y,z\in[0,1]$. Therefore, $I_{S,N,T}$ satisfies \textbf{\descref{(EP)}}.
\end{proof}

\begin{lem}\label{lem3.1}
Let $T$ be a $t$-norm, $N$ a strong fuzzy negation and $S$ a $t$-conorm such that $S$ is $N$-dual to $T$. Then, for all $x,y\in[0,1]$,
\begin{equation}\label{eq8}
I_{S,N,T}(N(x),y)=S(S(x,y),x)=S(S(x,x),y).
\end{equation}
\end{lem}

\begin{proof}
Straightforward.
\end{proof}

\begin{prop}\label{prop3.5}
If $S=S_{max}$ (maximum $t$-conorm), $T=T_{min}$ (minimum $t$-norm) and $N$ a strong fuzzy negation then $I_{S,N,T}$ satisfies \textbf{\descref{(L-CP)}}($N$).
\end{prop}

\begin{proof}
Since $S$ is $N$-dual to $T$ then, from Lemma \ref{lem3.1}, it follows that
\begin{eqnarray*}
I_{S,N,T}(N(x),y) &\overset{\eqref{eq8}}{=}& S_{max}(S_{max}(x,x),y) = S_{max}(x,y) \overset{\ref{itmS1}}{=} S_{max}(y,x) \\
&=& S_{max}(S_{max}(y,y),x) \overset{\eqref{eq8}}{=} I_{S,N,T}(N(y),x),
\end{eqnarray*}
for all $x,y\in[0,1]$. Therefore $I_{S,N,T}$ satisfies \textbf{\descref{(L-CP)}}($N$).
\end{proof}

The following results present methods of how to obtain $t$-norms and $t$-conorms from an ($S$,$N$,$T$)-implication and a fuzzy negation.

\begin{teo}\label{thm3.2}
Let $I_{S,N,T}$ be an ($S$,$N$,$T$)-implication and $N'$ a fuzzy negation. Define the function $S_{I_{S,N,T}}^{N'}:[0,1]^{2}\rightarrow[0,1]$ by
\begin{equation}\label{eq9}
S_{I_{S,N,T}}^{N'}(x,y)=I_{S,N,T}(N'(x),y)
\end{equation}
Then, it holds that:
\begin{enumerate}[label={\textbf{(\roman*)}}, ref=\textbf{(\roman*)}, align=left, leftmargin=*, noitemsep]
\item $S_{I_{S,N,T}}^{N'}(1,x)=S_{I_{S,N,T}}^{N'}(x,1)=1$ for all $x\in[0,1]$;
\item $S_{I_{S,N,T}}^{N'}$ is non-decreasing in both variables, i.e., $S_{I_{S,N,T}}^{N'}$ satisfies \ref{itmS3};
\item If $S=S_{max}$, $T=T_{min}$ and $N=N'$ is strong then $S_{I_{S,N,T}}^{N'}$ satisfies \ref{itmS1};
\item If $S=S_{max}$, $T=T_{min}$ and $N=N'$ is strong then $S_{I_{S,N,T}}^{N'}$ satisfies \ref{itmS4};
\item If $S=S_{max}$, $T=T_{min}$ and $N=N'$ is strong then $S_{I_{S,N,T}}^{N'}$ satisfies \ref{itmS2}.
\end{enumerate}
\end{teo}

\begin{proof}
It follows that:\vspace{0.2cm}
\\
\textbf{(i)} $S_{I_{S,N,T}}^{N'}(1,x)=I_{S,N,T}(N'(1),x)=I_{S,N,T}(0,x)\overset{\ref{itmI6}}{=}1$ and $S_{I_{S,N,T}}^{N'}(x,1)=I_{S,N,T}(N'(x),1)\overset{\ref{itmI7}}{=}1$.
\vspace{0.2cm}
\\
\textbf{(ii)} Let $x\leq z$ and $y\leq w$. By \ref{itmN2}, one has that $N'(x)\geq N'(z)$. Thereby,
\begin{eqnarray*}
S_{I_{S,N,T}}^{N'}(x,y) &=& I_{S,N,T}(N'(x),y) \underset{\ref{itmI2}}{\overset{\ref{itmI1}}{\leq}} I_{S,N,T}(N'(z),w) = S_{I_{S,N,T}}^{N'}(z,w).
\end{eqnarray*}
\textbf{(iii)} By Proposition \ref{prop3.5}, one has that $I_{S,N,T}$ satisfies \textbf{\descref{(L-CP)}}($N'$). Therefore,
$$S_{I_{S,N,T}}^{N'}(x,y)=I_{S,N,T}(N'(x),y)=I_{S,N,T}(N'(y),x)=S_{I_{S,N,T}}^{N'}(y,x).$$
\textbf{(iv)} By \textbf{(iii)}, $S_{I_{S,N,T}}^{N'}$ is commutative, and by Proposition \ref{prop3.2}\textbf{$(i)$}, $I_{S,N,T}$ satisfies \textbf{\descref{(NP)}}. Hence, it follows that $$S_{I_{S,N,T}}^{N'}(x,0)\overset{\textbf{(iii)}}{=}S_{I_{S,N,T}}^{N'}(0,x)=I_{S,N,T}(N'(0),x)=I_{S,N,T}(1,x)\overset{\textbf{\descref{(NP)}}}{=}x.$$
\textbf{(v)} By Proposition \ref{prop3.5}, one has that $I_{S,N,T}$ satisfies \textbf{\descref{(L-CP)}}($N'$) and by \textbf{(iii)}, $S_{I_{S,N,T}}^{N'}$ is commutative. Moreover, as $T$ is $N$-dual to $S$ then by Proposition \ref{prop3.4}, one has that $I_{S,N,T}$ satisfies \textbf{\descref{(EP)}}. Hence, we obtain that
\begin{eqnarray*}
S_{I_{S,N,T}}^{N'}(x,S_{I_{S,N,T}}^{N'}(y,z)) &=& I_{S,N,T}(N'(x),I_{S,N,T}(N'(y),z)) \\
&\overset{\textbf{\descref{(L-CP)}}\text{($N'$)}}{=}& I_{S,N,T}(N'(x),I_{S,N,T}(N'(z),y)) \\ &\overset{\textbf{\descref{(EP)}}}{=}& I_{S,N,T}(N'(z),I_{S,N,T}(N'(x),y)) \\
&\overset{\textbf{(iii)}}{=}& S_{I_{S,N,T}}^{N'}(I_{S,N,T}(N'(x),y),z) \\
&=& S_{I_{S,N,T}}^{N'}(S_{I_{S,N,T}}^{N'}(x,y),z).
\end{eqnarray*}
Therefore, $S_{I_{S,N,T}}^{N'}$ satisfies \ref{itmS2}.
\end{proof}

\begin{coro}\label{cor3.2}
Let $I_{S,N,T}$ be an ($S$,$N$,$T$)-implication and $N'$ be a fuzzy negation. If $S=S_{max}$, $T=T_{min}$ and $N=N'$ is strong then $S_{I_{S,N,T}}^{N'}$ is a t-conorm.
\end{coro}

\begin{coro}\label{cor3.3}
Let $I_{S,N,T}$ be an ($S$,$N$,$T$)-implication and $N'$ be a fuzzy negation. If $S=S_{max}$, $T=T_{min}$ and $N=N'$ is strong then the function $\widetilde{T}:[0,1]^{2}\rightarrow[0,1]$ defined by $$\widetilde{T}(x,y)=1-I_{S,N,T}(N'(1-x),1-y)$$ is a t-norm.
\end{coro}

\begin{proof}
The proof follows directly from Corollary \ref{cor3.2} and Proposition \ref{prop2.1}.
\end{proof}

Analogously we can prove the following

\begin{teo}\label{thm3.3}
Let $I_{S,N,T}$ be an ($S$,$N$,$T$)-implication and $N'$ a fuzzy negation. Define the function $T_{I_{S,N,T}}^{N'}:[0,1]^{2}\rightarrow[0,1]$ by
\begin{equation}\label{eq10}
T_{I_{S,N,T}}^{N'}(x,y)=N'(I_{S,N,T}(x,N'(y)))
\end{equation}
Then, it holds that:
\begin{enumerate}[label={\textbf{(\roman*)}}, ref=\textbf{(\roman*)}, align=left, leftmargin=*, noitemsep]
\item $T_{I_{S,N,T}}^{N'}(0,x)=T_{I_{S,N,T}}^{N'}(x,0)=0$ for all $x\in[0,1]$;
\item $T_{I_{S,N,T}}^{N'}$ is non-decreasing in both variables, i.e., $T_{I_{S,N,T}}^{N'}$ satisfies \ref{itmT3};
\item If $S=S_{max}$, $T=T_{min}$ and $N=N'$ is strong then $T_{I_{S,N,T}}^{N'}$ satisfies \ref{itmT1};
\item If $S=S_{max}$, $T=T_{min}$ and $N=N'$ is strong then $T_{I_{S,N,T}}^{N'}$ satisfies \ref{itmT4};
\item If $S=S_{max}$, $T=T_{min}$ and $N=N'$ is strong then $T_{I_{S,N,T}}^{N'}$ satisfies \ref{itmT2}.
\end{enumerate}
\end{teo}

\begin{coro}\label{cor3.4}
Let $I_{S,N,T}$ be an ($S$,$N$,$T$)-implication and $N'$ be a fuzzy negation. If $S=S_{max}$, $T=T_{min}$ and $N=N'$ is strong then $T_{I_{S,N,T}}^{N'}$ is a t-norm.
\end{coro}

\begin{coro}\label{cor3.5}
Let $I_{S,N,T}$ be an ($S$,$N$,$T$)-implication and $N'$ be a fuzzy negation. If $S=S_{max}$, $T=T_{min}$ and $N=N'$ is strong then the function $\widetilde{S}:[0,1]^{2}\rightarrow[0,1]$ defined by $$\widetilde{S}(x,y)=1-I_{S,N,T}(1-x,N'(1-y))$$ is a t-conorm.
\end{coro}

\begin{proof}
The proof follows directly from Corollary \ref{cor3.4} and Proposition \ref{prop2.1}.
\end{proof}


\section{Conclusions and future work}

In this paper we introduce a new class of fuzzy implications, which we call ($S$,$N$,$T$)-implications. Some of the main properties attributed to the fuzzy implications were investigated for this new class and some characterizations were presented. We show some relations between ($S$,$N$), ($T$,$N$) and ($S$,$N$,$T$)-implications for a given triple ($T$,$S$,$N$). In addition, we present methods on how to obtain $t$-norms and $t$-conorms from an ($S$,$N$,$T$)-implication and a fuzzy negation.

Among the proposals for future research we highlight the following: complement the study of properties that characterize the class of ($S$,$N$,$T$)-implications; provide a characterization of ($S$,$N$,$T$)-implications that satisfy important properties such as $I(x,I(y,z))=I(I(x,y),I(x,z))$ \cite{Cruz2018} and the law of importation; study fuzzy subsethood and entropy measures \cite{Dimuro2017,Pinheiro2018,Santos2019} generated from ($S$,$N$,$T$)-implications and introduce a generalization of this class to triples of the type ($G$,$N$,$O$) where $O$ is an overlap function, $N$ is a fuzzy negation and $G$ is a grouping function; and study their respective properties analogously to done in \cite{Dimuro2014,Dimuro2015,Dimuro2017,Dimuro2019,Dimuro2019a} for other classes of fuzzy implications based on overlaps and grouping functions.


\bibliographystyle{acm.bst}
\bibliography{refs}

\end{document}